\RequirePackage{amsmath}
\documentclass[runningheads]{llncs}
\usepackage[T1]{fontenc}
\usepackage{graphicx}
\usepackage{amssymb}
\usepackage{amsmath}

\DeclareMathOperator{\per}{per}
\DeclareMathOperator{\sign}{sgn}

\usepackage{tikz}
\usepackage{xcolor}  
\usepackage{hyperref}  

\hypersetup{
    colorlinks=true,     
    linkcolor=blue,      
    citecolor=blue,       
    urlcolor=blue        
}

\begin{document}
\title{Permanent of bipartite graphs in terms of determinants}
%
%
\author{Surabhi Chakrabartty\inst{1} \and
Ranveer Singh\inst{2}}
\authorrunning{S. Chakrabartty and R. Singh}
%
\institute{
    Indian Institute of Science Education and Research Berhampur, Berhampur, India\\
    \email{surabhic20@iiserbpr.ac.in} \and
    Indian Institute of Technology Indore, Indore, India\\
    \email{ranveer@iiti.ac.in}
}

\maketitle              
\begin{abstract}
Computing the permanent of a $(0,1)$-matrix is a well-known $\#P$-complete problem. In this paper, we present an expression for the permanent of a bipartite graph in terms of the determinant of the graph and its subgraphs, obtained by successively removing rows and columns corresponding to vertices involved in vertex-disjoint $4k$-cycles. Our formula establishes a general relationship between the permanent and the determinant for any bipartite graph. Since computing the permanent of a biadjacency matrix is equivalent to counting the number of its perfect matchings, this approach also provides a more efficient method for counting perfect matchings in certain types of bipartite graphs.

\end{abstract}
\section{Introduction}
A graph $G=(V, E)$ consists of a vertex set $V=\{ v_1, \ldots, v_n\}$ and an edge set $E$. It is said to be simple if it does not contain loops or multiple edges between a pair of vertices, and it is said to be undirected if the edges have no direction associated with them. In this paper, we will be concerned with simple and undirected graphs only. Adjacency matrix $A(G)$ of a graph $G$ with $n$ vertices is an $n \times n$ matrix defined by $A(G)= (a_{ij})$ where
\begin{align}
a_{ij} = 
\begin{cases}
    1, & \text{if there is an edge between } v_i \text{ and } v_j, \\ 
    0, & \text{otherwise}.
\end{cases}  \notag
\end{align}

A matching in a graph is a subset of edges such that no two edges share a common vertex. A perfect matching is a matching that covers every vertex of the graph. It is evident that a graph with an odd number of vertices cannot have a perfect matching. While determining the existence of a perfect matching in a general graph can be done efficiently~\cite{ref_article12}, counting the number of perfect matchings is computationally very hard.

A graph is said to be bipartite if its vertex set can be divided into two disjoint sets, such that no two vertices within the same set are adjacent. Given a bipartite graph $G=(U, V, E)$, its biadjacency matrix is a $|U| \times |V|$ matrix defined by $B(G)=(b_{ij})$ where 
\begin{align}
b_{ij} = 
\begin{cases}
    1, & \text{if there is an edge between } u_i \in U \text{ and } v_j \in V, \\ 
    0, & \text{otherwise}.
\end{cases}  \notag
\end{align}

The determinant and permanent of an $n \times n$ matrix $M = (a_{ij})$ are defined as 
\[ 
\det (M) = \sum_{\sigma \in S_{n}} \sign(\sigma) \prod_{i=1}^{n} a_{i,\sigma(i)}, \;\; 
\]
\text{and}\;\; 
\[
\per(M) = \sum_{\sigma \in S_{n}} \prod_{i=1}^{n} a_{i,\sigma(i)}, 
\] 
respectively, where $S_{n}$ is the set of all permutations of $\{1, 2, \dots, n\}$, the sign function $\sign (\sigma)$ is 1, if $\sigma$ is even and -1 if  $\sigma$ is odd, respectively. While the determinant of any matrix can be computed in $\mathcal{O}(n^{2.37})$, computing the permanent is known to be $\#P$-complete even for a $(0,1)$-matrix~\cite{ref_article1}. The determinant and permanent of a graph are defined as the determinant and permanent associated with its adjacency matrix, respectively. We denote them by $\det(G)$ and $\per(G)$.

The Pólya Permanent Problem~\cite{ref_article7} asks under what conditions the signs of certain nonzero entries in a $(0,1)$-matrix $A$ can be adjusted to obtain a matrix $B$ such that $\per(A) =  \det(B)$. This problem is known to be equivalent to twenty-three other graph-theoretic problems~\cite{ref_article6,ref_robertson1999}, including the problem of counting the number of perfect matchings in bipartite graphs. For a detailed survey and a complete characterization of minor-closed graph classes where perfect matchings can be computed efficiently, see~\cite{ref_killingvortex}. While computing the permanent of a general matrix $M$ is intractable in general, when $M$ is the biadjacency matrix of a bipartite planar graph, the permanent can be computed in polynomial time by reducing it to the Pfaffian of a suitably constructed skew-symmetric matrix, as established by Kasteleyn~\cite{ref_kasteleyn}.

Beyond its role in graph theory, the permanent also has important applications in quantum computing~\cite{ref_article5}. In Boson sampling, it helps determine the probability amplitudes of photon distributions in linear optical networks. It also shows up in statistical physics, particularly in models like the dimer model and monomer-dimer model, which are used to study phase transitions and other physical phenomena.

Writing the permanent as a function of the determinant has been a longstanding problem, more generally referred to as the Permanent vs. Determinant problem. In this paper, we aim to establish a relationship between the permanent and determinant of bipartite graphs using graph-theoretic concepts. As a corollary of this relation, we also observe the fact that the absence of $4k$-cycles in a bipartite graph guarantees that the permanent and determinant have the same absolute value. This formulation will not only provide insights into their general relationship but will also facilitate more efficient computation of the permanent for certain types of bipartite graphs.  

A cycle of length $l$ is a sequence of distinct vertices $u_1, u_2, \dots, u_l$ such that $u_i \sim  u_{i+1}$ for $i=1, \dots, l-1$ and $u_l\sim u_1$. A $4k$-cycle is a cycle with a length that is a multiple of four. Likewise, a $(4k+2)$-cycle is a cycle whose length exceeds a multiple of four by two. Since bipartite graphs do not contain odd-length cycles, the only cycles they can contain are $4k$-cycles and $(4k+2)$-cycles. We denote the set of all $4k$-cycles and $(4k+2)$-cycles of a graph $G$ by $C_{4k}(G)$ and $C_{4k+2}(G)$,  respectively.

We use the notation $T_z$ to represent any ordered tuple of $z$ mutually vertex-disjoint $4k$-cycles. Note that $T_1$ simply denotes any $4k$-cycle in the graph, as a single cycle is always vertex-disjoint by definition.

\begin{figure}[ht]
    \centering
    \begin{tikzpicture}

        \node[draw, circle, fill=black, inner sep=1.5pt, label=below:1] (1) at (0, 0) {};
        \node[draw, circle, fill=black, inner sep=1.5pt, label=below:2] (2) at (2, 0) {};
        \node[draw, circle, fill=black, inner sep=1.5pt, label=below:3] (3) at (3, 1) {};
        \node[draw, circle, fill=black, inner sep=1.5pt, label=left:4] (4) at (-1, 1) {};
        \node[draw, circle, fill=black, inner sep=1.5pt, label=above:5] (5) at (0, 2) {};
        \node[draw, circle, fill=black, inner sep=1.5pt, label=above:6] (6) at (2, 2) {};

        \draw[thick] (1) -- (2);
        \draw[thick] (2) -- (3);
        \draw[thick] (3) -- (4);
        \draw[thick] (4) -- (5);
        \draw[thick] (5) -- (6);
        \draw[thick] (4) -- (1);
        \draw[thick] (3) -- (6); 

        \node[draw, circle, fill=black, inner sep=1.5pt, label=below:7] (7) at (5, 0) {};
        \node[draw, circle, fill=black, inner sep=1.5pt, label=right:8] (8) at (6, 1) {};
        \node[draw, circle, fill=black, inner sep=1.5pt, label=above:9] (9) at (5, 2) {};
        \node[draw, circle, fill=black, inner sep=1.5pt, label=below:10] (10) at (4, 1) {};

        \draw[thick] (7) -- (8);
        \draw[thick] (8) -- (9);
        \draw[thick] (9) -- (10);
        \draw[thick] (10) -- (7);

        \draw[thick] (6) -- (9);

    \end{tikzpicture}
    \caption{Example of a bipartite graph}
    \label{fig:example-figure}
\end{figure}
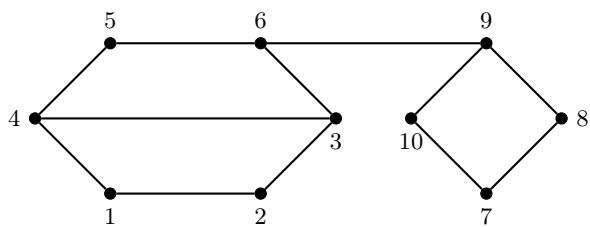
\begin{example}
     Consider the graph $G$ shown in Figure \ref{fig:example-figure}. The cycles in $G$ are:
     \begin{equation}
     C_{1}=(1, 2, 3, 4, 1), C_{2}=(3, 4, 5, 6, 3), C_{3}=(7, 8, 9, 10, 7),  
     C_{4}=(1, 2, 3, 6, 5, 4, 1), \notag
     \end{equation}
    Then the set of $4k$-cycles in $G$ is given by 
    \[
        C_{4k}(G) = \{C_{1}, C_{2}, C_{3} \},
    \]
    and the set of $(4k+2)$-cycles is 
    \[
        C_{4k+2}(G) = \{C_{4}\}.
    \]
    Now, the single-cycle tuples are:
    \begin{align}
        T_1^{1} = C_1, \quad T_1^{2} = C_2, \quad T_1^{3} = C_3. \notag
    \end{align}
    The two-cycle vertex-disjoint tuples are:
    \begin{align}
        T_2^{1} = (C_1, C_3), \quad T_2^{2} = (C_2, C_3),  \quad  T_2^{3} = (C_3, C_1), \quad T_2^{4} = (C_3, C_2). \notag
    \end{align}
\end{example}

The rest of the paper is organized as follows: In Section \ref{Section_2}, we give the graph-theoretic interpretation of the permanent by introducing Sachs subgraphs. In the main result (Section \ref{Section_3}), we demonstrate how the determinant of the submatrices obtained after removing rows and columns corresponding to $4k$-cycles can be used to compute the permanent of a bipartite graph. Finally, in the discussion (Section \ref{Section_4}), we provide a few conditions such that the permanent of such a type of bipartite graph can be computed efficiently.

\section{Permanent using Sachs Subgraphs} \label{Section_2}
A Sachs subgraph is a subgraph in which every component is either a cycle or an individual edge. We observe that, for a bipartite graph, there is no Sachs subgraph on an odd number of vertices, as bipartite graphs do not contain odd-length cycles. We will use the notation 
$U_i$ to denote any Sachs subgraph of $G$ on $i$ vertices.

In the graph shown in Figure \ref{fig:example-figure}, some possible Sachs subgraphs on six vertices are 
    \begin{align}
        U_{6}^{1} &= \{C_4\}, \quad U_{6}^{2} = \{C_{1} \cup (5, 6)\}, \quad U_{6}^{3} = \{C_{1} \cup (7, 8)\}, \notag \\
        U_{6}^{4} &= \{C_{2} \cup (1, 2)\}, \quad U_{6}^{5} = \{C_{3} \cup (3, 4)\}, \notag \\
        U_{6}^{6} &= \{(1, 4) \cup (6, 9) \cup (7, 8)\}, \quad 
        U_{6}^{7} = \{(1, 2) \cup (3, 4) \cup (5, 6)\}. \notag
    \end{align}
In total, there are 94 Sachs subgraphs of $G$ on six vertices.

Some possible Sachs subgraphs of $G$ on four vertices are
\begin{align}
U_{4}^{1} &= \{C_{1}\}, \quad U_{4}^{2} = \{C_{2}\}, \quad U_{4}^{3} = \{C_{3}\}. \notag \\
U_{4}^{4} &= \{(2, 3) \cup (9, 10)\}, \quad U_{4}^{5} = \{(3, 6) \cup (1, 2)\}. \notag 
\end{align}
In total, there are 51 Sachs subgraphs of $G$ on four vertices.
Furthermore, each Sachs subgraph of $G$ on two vertices contains exactly one edge of the graph $G$. 

Hence, from this example, we observe that in a bipartite graph, the existence of a Sachs subgraph on $i$ vertices, implies the existence of Sachs subgraph on $2, 4, \ldots, i-2$ vertices. This is because if the Sachs subgraph has even one cycle, then we can decompose it to disjoint edges or copies of $K_2$ by removing alternate edges, and if the Sachs subgraph has no cycles and comprises edges only, then we can remove the edges one by one to obtain smaller Sachs subgraphs. We will use this fact for the construction of different Sachs subgraphs from the known ones.

\subsection{Permanent using the determinant of subgraphs}

We recall the formula for computing the determinant and permanent using Sachs subgraphs in the following way  
~\cite{ref_book1,ref_article2}:
\begin{equation}
    \det(G) = \sum_{U_{n}}(-1)^{n-p(U_{n})}2^{c(U_{n})} \text{ and } \per(G)= \sum_{U_{n}} 2^{c(U_{n})}, \label{eq:0.1}
\end{equation}
where the summation is taken over all the Sachs subgraphs $U_{n}$ of $G$ on $n$ (where $n=|V|$) vertices, $p(U_{n})$ denotes the number of components in $U_{n}$, and $c(U_{n})$ denotes the number of cycles in $U_{n}$.

As observed earlier, in a bipartite graph,  there is no Sachs subgraph on an odd number of vertices; thus, $\det(G)=\per(G)=0$ when $n$ is odd~\cite{ref_article8,ref_article9}. Thus, from now on, we will consider $n$ to be even. 

Our goal is to express the permanent of any bipartite graph $G$ in terms of determinants, thereby establishing a general relationship between the two.  Equation~\eqref{eq:0.1} can be written as
 \begin{align}
    \det(G)= \sum_{U_{n}} (-1)^{s+t+r} 2^{s+t} \text{ and } \per(G)= \sum_{U_{n}} 2^{s+t}, \label{eq:1.3}
\end{align}
respectively, where $U_{n}$ varies over the set of Sachs subgraphs of $G$ on $n$  vertices, $s$ denotes the number of $4k$ cycles in $U_n$, $t$ denotes the number of $(4k+2)$-cycles in $U_{n}$ and $r$ denotes the number of edges in $U_{n}$. That is, $U_n$ can be expressed in the following way: 
\[
U_{n}= \{ C_{4k_{1}}\cup \dots \cup C_{4k_{s}} \} \cup \{C_{4l_{1}+2}\cup \dots  \cup C_{4l_{t}+2} \} \cup  \underbrace{ \{K_2 \cup \dots \cup K_2\} }_{r \text{ times}}, 
\]
for some positive integers $k_{1}, \dots, k_{s}, l_{1}, \dots, l_{t}$. We use the notation $K_2$ to represent any edge; however, it is important to note that in our case, each $K_2$ is distinct and does not share any vertices with another $K_2$ in the same Sachs subgraph.  
 As $U_n$ represents a Sachs subgraph on $n$ vertices, we must have the following equality:
\begin{align}
    &n= 4(k_{1} + \dots + k_{s}) + 4(l_{1} + \dots + l_{t}) + 2(t+r).  \notag 
\end{align}
As we have taken $n$ to be even, we have
\begin{align}
    \frac{n}{2} \equiv t + r\pmod{2}. \notag
\end{align} 
  Therefore, using Equation~\eqref{eq:1.3}, determinant reduces to 
\begin{align}
    \det(G) = \sum_{U_{n}} (-1)^{s+\frac{n}{2}} 2^{s+t}. \label{eq:1.4} 
\end{align}

\subsubsection{\textbf{Determinant of subgraphs using $4k$-cycles}}

 We will use the notation $\det(G \backslash R)$ to denote the determinant of the submatrix obtained by removing the rows and columns corresponding to the vertices in the $4k$-cycle $R$ of $G$. Now, we use Formula~\eqref{eq:1.4} to obtain:
\begin{align}
    \sum_{R_{1} \in C_{4k}(G)}  \det(G \backslash R_{1}) 
    = \sum_{R_{1} \in C_{4k}(G)}   \sum_{W_{n - |V(R_1)|}} (-1)^{s^{\prime} + \frac{n - |V(R_1)|}{2}} 2^{s^{\prime}+t^{\prime}} , \notag 
\end{align}
where the last summation is over all the Sachs subgraph $W_{n - |V(R_1)|}$ of $G\backslash R_{1}$, $s^{\prime}$ denotes the number of $4k$ cycles in $W_{n - |V(R_1)|}$ and $t^{\prime}$ denotes the number of $(4k+2)$-cycles in $W_{n - |V(R_1)|}$. 

Using the fact that $(-1)^{\frac{n - |V(R_1)|}{2}} = (-1)^{\frac{n}{2}}$ (since $|V(R_1)|$ is a multiple of 4), and the one-to-one correspondence between the Sachs subgraph $U_n$ in $G$ containing $R_1$ and the Sachs subgraph $W_{n - |V(R_1)|}$ in $G \setminus R_1$, and noting that $s = s' + 1$ (since $s$ represents the number of $4k$-cycles in $U_n$ and always includes one more—namely, $R_1$—than the number $s'$ in $W_{n - |V(R_1)|}$), we obtain
\begin{align}
    \sum_{R_{1} \in C_{4k}(G)} \det(G \backslash R_{1}) 
    = \sum_{R_{1} \in C_{4k}(G)} \;\;\; \sum_{\substack{U_{n} \text{ containing} \\ R_{1}}} (-1)^{s-1+\frac{n}{2}} 2^{s-1+t}. \label{eq:1.8}
\end{align} 
We will iteratively apply Equation~\eqref{eq:1.8} by using two disjoint $4k$-cycles, then three, then four, and so on, while computing the determinant for each resulting reduced subgraph.

\section{The main result} \label{Section_3}

\begin{theorem}\label{theorem1}
Let $G$ be a bipartite graph with $n$ vertices. If $G$ contains at most $m$ vertex-disjoint $4k$-cycles then
    \begin{align}
            \per(G) = \begin{cases}
        (-1)^{n/2} \displaystyle\sum_{z=0}^{m} \frac{4^z}{z!} \sum_{T_z} \det(G \setminus T_z), & \text{if } n \text{ is even}, \\
        0, & \text{if } n \text{ is odd}, 
    \end{cases} \label{eq:main_theorem}
    \end{align}
where the sum varies over all possible ordered tuples $T_z$ of $z$ mutually vertex-disjoint $4k$-cycles in $G$.  
\end{theorem}

\begin{proof} For a bipartite graph with an odd number of vertices, the permanent is zero, as follows trivially from our earlier observations about the non-existence of a Sachs subgraph on an odd number of vertices in a bipartite graph. Therefore, it remains to prove the case when the number of vertices is even.

Recall from Equation~\eqref{eq:1.8} that
\begin{align}
 \sum_{T_1}  \det(G \backslash T_1) 
 =&
    \sum_{R_{1} \in C_{4k}(G)} \det(G \backslash R_{1}) \notag \\ 
     =& \sum_{R_{1} \in C_{4k}(G)} \sum_{\substack{U_{n} \text{ containing} \\  R_{1}}} (-1)^{s-1+\frac{n}{2}} 2^{s+t-1}, \notag  
\end{align}
where the first summation on the right hand side, is taken over all $4k$-cycles $R_1$ of $G$, the second summation is over all Sachs subgraphs $U_n$ of $G$ on  $n$ vertices that contain $R_1$; $s$ denotes the number of $4k$-cycles in $U_n$, and $t$ denotes the number of $(4k + 2)$-cycles in $U_n$.
Notice that on the right-hand side of the equation, each Sachs subgraph is counted exactly as many times as the number of $4k$-cycles it contains.
Hence,
\begin{align}
     \sum_{T_1}  \det(G \backslash T_1) 
     = (-1)^{\frac{n}{2}} \sum_{j = 1, 2, 3, \dots, m} (-1)^{j-1} j 2^{j-1} \! \! \sum_{\substack{U_{n} \text{ containing} \\ \text{exactly  } j \; 4k \text{-cycles}}}  2^{t}. \notag 
\end{align}
Similarly, using Equation~\eqref{eq:1.8} again, we observe that for each $R_1 \in C_{4k}(G)$, we have
\begin{align}
    \sum_{R_{2} \in C_{4k}(G\backslash R_{1})}  & \det(G \backslash R_{1}\backslash R_{2}) \notag \\
    & =  (-1)^{\frac{n-|V(R_1)|}{2}} \sum_{\ R_{2} \in C_{4k}(G\backslash R_{1})} \sum_{\ \substack{W_{n - |V(R_1)|} \\ \text{containing } R_{2}}}  (-1)^{s^{\prime}-1}
    2^{s^{\prime}+t^{\prime}-1} ,
\end{align}
where the last summation is taken over all the Sachs subgraph $W_{n - |V(R_1)|}$  of $G\backslash R_{1}$ that contain $R_{2}$, $s^{\prime}$ denotes the number of $4k$-cycles in $W_{n - |V(R_1)|}$, and $t^{\prime}$ denotes the number of $(4k+2)$-cycles in $W_{n - |V(R_1)|}$. 
Following the same trick as before, we note that in the above equation's right-hand side, each Sachs subgraph is summed exactly as many times as the number of $4k$-cycles it contains.
Hence,
\begin{align} 
     \sum_{T_2}  \det(G \backslash T_2) & = \sum_{R_{1} \in C_{4k}(G)}  \sum_{R_{2} \in C_{4k}(G\backslash R_{1})} \det(G \backslash R_{1}\backslash R_{2}) \notag  \\
    & = (-1)^{\frac{n}{2}}  \sum_{R_{1} \in C_{4k}(G)}  \sum_{j = 1, 2, \dots, m} (-1)^{j-1} j 2^{j-1} \sum_{\substack{W_{n - |V(R_1)|} \text{containing} \\ \text{exactly  } j \; 4k \text{-cycles} }}  2^{t^{\prime}} \notag \\
    & = (-1)^{\frac{n}{2}} \sum_{R_{1} \in C_{4k}(G)} \sum_{W_{n - |V(R_1)|}} (-1)^{s^{\prime}-1} s^{\prime} 2^{s^{\prime}+t^{\prime}-1}. \notag 
\end{align}
Using the fact that there is a one-to-one correspondence between the Sachs subgraph $U_{n}$ in $G$ containing $R_{1}$ and the Sachs subgraph $W_{n - |V(R_1)|}$ in $G\backslash R_{1}$, and noting that $s = s' + 1$ (since $s$ represents the number of $4k$-cycles in $U_{n}$ containing $R_{1}$ and hence starts from $1$, always containing one $4k$-cycle more—namely, $R_{1}$—than $W_{n - |V(R_1)|}$), we obtain 
\begin{align}
     \sum_{T_2}  \det(G \backslash T_2)
    &  = (-1)^{\frac{n}{2}} \sum_{\ R_{1} \in C_{4k}(G)} \sum_{ \ \substack{U_{n} \text{ containing} \\  R_{1}}} (-1)^{s-2} (s-1) 2^{s+t-2} \notag \\ 
    & = (-1)^{\frac{n}{2}} \sum_{j = 2, 3, 4 \dots, m} (-1)^{j-2} (j-1) j 2^{j-2} \!\!\!\! \sum_{\substack{U_{n} \text{ containing} \\ \text{exactly  } j \; 4k \text{-cycles}}} 2^{t} \notag \\
     & = (-1)^{\frac{n}{2}} \sum_{j = 2, 3, 4 \dots, m} (-1)^{j-2} 2! \binom{j}{j-2} 2^{j-2} \!\!\!\! \sum_{\substack{U_{n} \text{ containing} \\ \text{exactly  } j \; 4k \text{-cycles}}} 2^{t}. \notag 
\end{align}
Similarly, it can be verified by mathematical induction on $z$ that for any $z \in \{ 0, 1, 2, \dots, m \}$,
\begin{align} 
    \sum_{T_z}  \det(G \backslash T_z) 
     &  = (-1)^{\frac{n}{2}} \sum_{j = z}^ {m}  (-1)^{j-z} z! \binom{j}{j-z} 2^{j-z} \sum_{\substack{U_{n} \text{\ containing} \\ \text{exactly } j \; 4k \text{-cycles}}} \!\!\!\!\! 2^{t}. \label{eq:3}
\end{align}
Next, for each value of $z \in \{0, 1, 2, \dots, m \}$, we multiply Equation~\eqref{eq:3}  by the coefficient $\frac{4^z}{z!}$ and sum over all values of $z$ to obtain 
\begin{equation}
\begin{aligned}
    \sum_{z=0}^{m} \frac{4^z}{z!} \sum_{T_z } \det(G \setminus T_z) 
    &=  (-1)^{\frac{n}{2}} \sum_{z=0}^{m}  \frac{4^z}{z!}    \sum_{j = z}^ {m} (-1)^{j-z} z! \binom{j}{j-z} \, 2^{j-z} \!\!\!\!\!\!\!\!  \sum_{\substack{ U_{n} \text{ containing} \\ \text{exactly } j \, 4k\text{-cycles}}}  \!\!\!\!\!\!\!\! 2^{t}   \notag \\
    &  = (-1)^{\frac{n}{2}}  \sum_{z=0}^{m}   \sum_{j = z}^ {m} (-1)^{j-z} \binom{j}{j-z} \, 2^{j+z}  \!\!\!\!\!\!\!\! \sum_{\substack{ U_{n} \text{ containing} \\ \text{exactly } j \, 4k \text{-cycles}}}  \!\!\!\!\!\!\!\! 2^{t}   \notag \\
    & = (-1)^{\frac{n}{2}} \sum_{j = 0}^ {m} 2^{j} \sum_{z=0}^{j}  (-1)^{j-z} \binom{j}{j-z} 2^z \!\!\!\!  \sum_{\substack{ U_{n} \text{ containing} \\ \text{exactly } j \, 4k \text{-cycles}}} \!\!\!\!\!\!\! 2^{t}  . \notag  
\end{aligned}
\end{equation}
By using the binomial theorem, we have:
\begin{align}
    \sum_{z=0}^{j} (-1)^{j-z} \binom{j}{j-z} 2^z & = (-1+2)^j  \notag \\
    & = 1. \notag 
\end{align}
Hence, 
\begin{equation}  
\begin{aligned}
      \sum_{z=0}^{m} \frac{4^z}{z!} \sum_{T_z } \det(G \setminus T_z) 
    &   = (-1)^{\frac{n}{2}} \sum_{j = 0}^ {m} 2^{j}  \sum_{\substack{ U_{n} \text{ containing} \\ \text{exactly  } j \; 4k \text{-cycles}}} 2^{t}   \notag \\
    & = (-1)^{\frac{n}{2}} \sum_{U_{n}} 2^{s+t}\notag \\
    & = (-1)^{\frac{n}{2}} \per(G). \notag
\end{aligned}
\end{equation}
This completes the proof. \qed 
\end{proof}

\begin{corollary}
    If a bipartite graph $G$ is free of $4k$-cycles then
    \begin{equation}
    \per(G)= (-1)^\frac{n}{2} \det(G).
    \end{equation}
\end{corollary}
In simpler words, the absence of 
$4k$-cycles in a bipartite graph guarantee that the permanent and determinant have the same absolute value.

\begin{theorem} \label{theorem2}
    Let $G$ be a bipartite graph and let $G_i$ denote an induced subgraph of $G$ on $i$ vertices. Then $G$ contains at most $m$ vertex-disjoint $4k$-cycles if and only if 
    \begin{align}
            \per(G_i) = (-1)^{i/2} \sum_{z=0}^{m} \frac{4^z}{z!} \sum_{T_z } \det(G_i \setminus T_z), \label{eq:2}
    \end{align}
 holds for all such $G_i$, where \( i \in \{0, 2, 4, \ldots \} \), and the inner sum varies over all possible ordered tuples $T_z$ of $z$ mutually vertex-disjoint $4k$-cycles in $G_i$.
 \end{theorem}
 
\begin{proof}
Since for each even \( i \), \( G_i \) is itself a bipartite graph on \( i \) vertices, the forward implication follows directly from Theorem~\ref{theorem1}.

Let us prove the backward implication. Given that Equation~\eqref{eq:2} holds,  we must show that $G$ does not contain more than  $m$ vertex-disjoint $4k$-cycles. It is equivalent to showing that none of the Sachs subgraphs of $G$ contains more than $m$ $4k$-cycles.

Let us assume that a Sachs subgraph contains at least $(m+1)$ $4k$-cycles. Now, by keeping any $m+1$ $4k$-cycles of this Sachs subgraph and removing the alternate edges from the other $4k$-cycles, we can derive a Sachs subgraph containing exactly $m+1$ $4k$-cycles.

For all $i \in \{0, 2, 4, \dots\}$, we have:
\begin{align}
    \per(G_i) = (-1)^{i/2} \sum_{z=0}^{m} \frac{4^z}{z!} \sum_{T_z } \det(G_i \setminus T_z).
\end{align}
Rewriting in terms of Sachs' subgraphs, we have
\begin{align}
    \sum_{j = 0}^{n} 2^{j} \sum_{\substack{ U_{i} \text{ containing} \\ \text{exactly  } j \; 4k \text{-cycles}}} \!\!\!\! 2^{t}
    & = \sum_{z=0}^{m} \frac{4^z}{z!} \; \; \sum_{j = z}^{n}  (-1)^{j-z} z! \binom{j}{j-z} \, 2^{j-z} \!\!\!\!\!\!\!\!\!\! \sum_{\substack{ U_{i} \text{ containing} \\ \text{exactly } j \, 4k \text{-cycles}}} \!\!\!\!\!\!\!\!\!\! 2^{t}  \notag \\  
    & =\sum_{j = 0}^{n} 2^{j} \sum_{z=0}^{\min(j, m)} (-1)^{j-z} \binom{j}{j-z} 2^z \sum_{\substack{ U_{i} \text{ containing} \\ \text{exactly } j \, 4k \text{-cycles}}} \!\!\!\!\!\! 2^{t}.  \label{eq:6} 
\end{align}
For $j \leq m$, we use the binomial theorem to obtain
\begin{align}
    \sum_{z=0}^{j} (-1)^{j-z} \binom{j}{j-z} 2^z 
    & = (-1+2)^{j} \notag \\
    & = 1. \label{eq:7}
\end{align}
For $j=m+1$, the binomial theorem results in
\begin{align}
    \sum_{z=0}^{m} (-1)^{m+1-k} \binom{m+1}{m+1-k} 2^z &= (-1+2)^{m+1} -2^{m+1} \notag \\
    & = 1-2^{m+1}. \label{eq:7.5}
\end{align}
We substitute Equation~(\ref{eq:7}) into Equation~(\ref{eq:6}) to obtain:
\begin{align}
    \sum_{j = 0}^{n} 2^{j} \sum_{\substack{ U_{i} \text{ containing} \\ \text{exactly  } j \; 4k \text{-cycles}}} \!\!\!\! 2^{t}   &= \;\;\; \sum_{j = 0}^{m} 2^{j} \sum_{\substack{ U_{i} \text{ containing} \\ \text{exactly } j \, 4k \text{-cycles}}} \!\!\!\!\!\! 2^{t} \notag \\
    & \quad +  \sum_{j = m+1}^{n} 2^{j} \sum_{z=0}^ {m}  (-1)^{j-z} \binom{j}{j-z} 2^z \!\!\!\!\!\! \sum_{\substack{ U_{i} \text{ containing} \\ \text{exactly } j \, 4k \text{-cycles}}} \!\!\!\!\!\! 2^{t}. \notag 
\end{align}
On simplification, we have
\begin{align}
    \sum_{j = m+1}^{n}  2^{j} (1 - \sum_{z=0}^ {m}  (-1)^{j-z} \binom{j}{j-z} 2^z) \sum_{\substack{ U_{i} \text{ containing} \\ \text{exactly } j \,  4k \text{-cycles}}} \!\!\!\! 2^{t} = 0. \notag 
\end{align}
Substituting the value from Equation~\eqref{eq:7.5}, we obtain  
\begin{align}
      \sum_{j = m+2}^{n}  2^{j} (1 - \sum_{z=0}^ {m}  &(-1)^{j-z} \binom{j}{j-z} 2^z) \!\!\!\!\!\!\!\! \sum_{\substack{ U_{i} \text{ containing} \\ \text{exactly } j \,  4k \text{-cycles}}} \!\!\!\!\!\!\!\! 2^{t} \;\; + \;\; 2^{2m+2} \!\!\!\!\!\!\!\!\!\! \sum_{\substack{ U_{i} \text{ containing} \\ \text{exactly } m+1 \, 4k \text{-cycles}}} \!\!\!\!\!\!\!\! 2^{t}  = 0.  \label{eq:8}
\end{align}
From Equation~\eqref{eq:8}, it follows that there must exist a Sachs subgraph of $G$ with more than $m+1$ $4k$-cycles. If no such subgraph exists, the first sum in Equation~\eqref{eq:8} would be zero. Consequently, the last term must also be zero, as the right-hand side is zero. This contradicts the assumption that a Sachs subgraph exists with exactly $m+1$ $4k$-cycles. Hence, we conclude that a Sachs subgraph of $G$ with at least $m+2$ $4k$-cycles exists. By a similar construction as in our earlier arguments, this further implies the existence of a Sachs subgraph with exactly $m+2$ $4k$-cycles.
 
Among all Sachs subgraphs of $G$ containing exactly $m+2$ $4k$-cycles, we choose the Sachs subgraph containing the least number of vertices and call it $U_{q}^{\prime}$. 
 Also, observe that $U_{q}^{\prime}$ must not contain any $(4k+2)$-cycles or edges due to minimality of $q$. Hence, 
\[
U_{q}^{\prime}= C_{1} \cup C_{2} \cup \dots \cup C_{m+2},
\]
for some vertex-disjoint $C_{1}, C_{2}, \dots, C_{m+2} \in C_{4k}(G)$. 

Let $p = \mid V(C_{1}) \mid + \mid V(C_{2})  \mid + \dots + \mid V(C_{m+1})  \mid < q$. We know that Equation~\eqref{eq:8} is satisfied for all $i \in \{0, 2, 4, \dots \}$, so it should be satisfied for $i=p$ as well. As $p<q$, there is no Sachs subgraph of $G$ on $p$ vertices which contains more than $m+1$ $4k$-cycles. Hence, by setting $i=p$ in Equation~\eqref{eq:8}, we obtain
\begin{align}
    2^{2m+2} \sum_{\substack{U_{p} \text{containing} \\ \text{exactly } m+1 \   4k \text{-cycles}}} 2^{t} = 0. \notag
\end{align}
This is a contradiction to our assumption since there exists a Sachs subgraph $U_{p}^{\prime}= U_{q}^{\prime} \setminus C_{m+2} = C_{1} \cup C_{2} \cup \dots \cup C_{m+1}$ and hence the sum is non-zero. Thus, no Sachs subgraph of $G$ contains more than $m$ $4k$-cycles. Therefore, the backward implication of Theorem \ref{theorem2} holds.
This completes the proof.  \qed
\end{proof}

\begin{example}
Consider the bipartite graph $G$ in Figure \ref{fig:example-figure}. Since $C_1$ and $C_2$ share a common vertex, only one of them can be included at a time in the tuple of vertex-disjoint $4k$-cycles along with $C_3$; hence, the maximum number of vertex-disjoint $4k$-cycles in $G$ is 2. Therefore, the permanent of $G$ can be computed using Equation~\eqref{eq:main_theorem} with $m = 2$.

\begin{align}
     \per(G) &=  (-1)^{n/2} \sum_{z=0}^{2} \frac{4^z}{z!} \sum_{T_z } \det(G \setminus T_z) \notag \\
     & = (-1)^{n/2} \Bigg [ \det(G) + 4 \sum_{T_1} \det(G \setminus T_1) + \frac{4^{2}}{2!} \sum_{T_2} \det(G \setminus T_2) \Bigg]. \label{eq:example} 
\end{align}
as $n$ is even. Now, 
\begin{align}
    \det(G)=0. \notag
\end{align}
Recall that $C_{4k}(G) = \{C_{1}, C_{2},C_{3} \}$ . Hence, 
\begin{align}
    \sum_{T_1} \det(G \backslash T_1) 
    &= \sum_{R_{1} \in C_{4k}(G)} \det(G \backslash R_{1}) \notag \\
    &= \det(G \backslash C_{1}) + \det(G \backslash C_{2})+\det(G \backslash C_{3})  \notag \\
    & = 0 + 0 + (-1)\notag \\ 
    &   = -1. \notag
\end{align}
and,
    \begin{align}  \sum_{T_2}  \det(G \backslash T_2) =&
    \sum_{R_{1} \in C_{4k}(G)}  \sum_{R_{2} \in C_{4k}(G\backslash R_{1})} \det(G \backslash R_{1}\backslash R_{2} ) \notag \\
    &= \det(G \backslash C_{1} \backslash C_{3}) + \det(G \backslash C_{2} \backslash C_{3}) \notag \\ 
    & \quad + \det(G \backslash C_{3} \backslash C_{1}) + \det(G \backslash C_{3} \backslash C_{2}) \notag \\ 
    & =  (-1) +  (-1) + (-1) + (-1) \notag \\
    & = -4. \notag
\end{align}
Hence, using~\eqref{eq:example}, we get:
\begin{align}
     \per(G) & = (-1)^\frac{10}{2} \{ 0 + 4(-1)+ \frac{4^{2}}{2!}(-4) \}\notag \\
     & = 36. \notag
 \end{align}
\end{example}

\section{Discussion} \label{Section_4}

Theorem \ref{theorem1} can be utilized to compute the permanent for any bipartite graph. Since permanent is closely related to the count of perfect matching, this formula could also be used to compute the number of perfect matching of bipartite graphs, where the biadjacency matrix itself serves as the adjacency matrix of another bipartite graph. The formula consists of two summations. The first summation runs from 1 to $m$, where $m$ is the maximum number of vertex-disjoint $4k$-cycles. Since $m$ can be at most $n/4$, this summation contributes at most $n/4$ terms. However, the second summation is taken over all $z$-tuples of mutually vertex-disjoint $4k$-cycles, which can grow exponentially with the size of the graph, making it computationally challenging. Thus, the main challenges in this process include efficiently identifying all $4k$-cycles within the graph and ensuring that their number does not grow too large, as it directly impacts computational efficiency.

These challenges can be addressed by identifying classes of bipartite graphs that contain only a limited number of $ 4k$ cycles and allow efficient enumeration of these cycles. Trees, acyclic structures, and cycle graphs are particularly easy to handle using this formula due to their simpler structures. In general, enumerating all cycles in a graph is an NP-hard problem, but recent advancements~\cite{ref_article3,ref_article11} have shown that all short cycles up to length $2g-2$ in a bipartite graph can be enumerated efficiently. Therefore, it is worth investigating bipartite graphs that are characterized by the presence of only short cycles.

The bipartite cactus graph, a connected graph that is known for sharing at most one vertex between any two cycles, along with the condition:
\[g > \frac{n + \binom{c}{2}+c}{c + 2}, \]
where $g, n, c$ are the number of girth, the number of vertices, and the number of cycles of length $g$, respectively, is one such candidate where the theorem can be used efficiently. It would be valuable to identify more classes of bipartite graphs where this formula can be applied efficiently. 

\section*{Acknowledgment}
We thank Hitesh Wankhede, Noga Alon, and RB Bapat for their valuable suggestions. We also thank the anonymous reviewers for their insightful comments and acknowledge financial support from the DST INSPIRE grant.

%
%
\bibliographystyle{splncs04}
%

\end{document}